\documentclass[10pt,journal,compsoc]{IEEEtran}
\usepackage{ifpdf, flushend,subfigure}
\DeclareUnicodeCharacter{3000}{ }

\ifCLASSINFOpdf
  \usepackage[pdftex]{graphicx}
  \graphicspath{{../pdf/}{../jpeg/}}
  \DeclareGraphicsExtensions{.pdf,.jpeg,.png}
\else
  \usepackage[dvips]{graphicx}
  \graphicspath{{../eps/}}
  \DeclareGraphicsExtensions{}
\fi
\usepackage{epstopdf}
\usepackage{multirow}
\usepackage{float}
\usepackage[cmex10]{amsmath}
\usepackage {amssymb}
\usepackage{mathtools}
\usepackage{graphicx}
\usepackage{array}
\usepackage{mdwmath}
\usepackage{mdwtab}
\usepackage{eqparbox}
\usepackage{nomencl}
\usepackage{cite}
\usepackage{algorithm}
\usepackage{algorithmic}
\usepackage{makeidx}
\usepackage{ifthen}
\usepackage{subfigure}
\usepackage{caption}
\usepackage{amsthm}
\usepackage{amsfonts}
\usepackage{pdflscape}

\newcommand{\beq}{\begin{equation}}
\newcommand{\eeq}{\end{equation}}
\newcommand{\beqn}{\begin{eqnarray}}
\newcommand{\eeqn}{\end{eqnarray}}
\newcommand{\beqno}{\begin{eqnarray*}}
\newcommand{\eeqno}{\end{eqnarray*}}
\newcommand{\bma}{\begin{displaymath}}
\newcommand{\ema}{\end{displaymath}}
\newcommand{\bnu}{\begin{enumerate}}
\newcommand{\enu}{\end{enumerate}}
\newcommand{\bce}{\begin{center}}
\newcommand{\ece}{\end{center}}
\newcommand{\btb}{\begin{tabular}}
\newcommand{\etb}{\end{tabular}}

\theoremstyle{plain}
\newtheorem{theorem}{Theorem}[section]
\newtheorem{lemma}[theorem]{Lemma}

\newtheorem{corollary}[theorem]{Corollary}
\newtheorem{assumption}[theorem]{Assumption}
\theoremstyle{definition}

\newtheorem{remark}[theorem]{Remark}

\usepackage{physics}
\usepackage{braket}
 \usepackage{url}

\begin{document}

\title{\emph{QC-C$\&$CG}: Quantum-Classical Algorithm for Two-stage Adaptive Robust Optimization}
\author{\IEEEauthorblockN{Duong The Do,~\IEEEmembership{Student Member,~IEEE}, Jiaming Cheng,~\IEEEmembership{Student Member,~IEEE}, Duong Tung Nguyen,~\IEEEmembership{Member,~IEEE}}  
\thanks{Duong The Do, Jiaming Cheng, and Duong Tung Nguyen are with the School of Electrical, Computer and Energy Engineering, Arizona State University, Tempe, AZ, United States.
Email: \textit\{duongdo, jiaming, duongnt\}@asu.edu.
(\textit{Corresponding author}: Duong The Do).
}}

\IEEEtitleabstractindextext{
\begin{abstract}
   Quantum optimization provides a promising approach for solving large-scale combinatorial problems through quadratic unconstrained binary optimization (QUBO) formulations. However, integrating QUBO-based solvers into structured optimization frameworks while preserving solution guarantees remains a fundamental challenge. This paper develops a hybrid quantum-classical column-and-constraint generation (QCCG) framework for solving two-stage adaptive robust optimization problems with binary first-stage decisions and linear recourse under polyhedral uncertainty. The proposed approach reformulates the restricted master problem as a QUBO and solves it approximately using a quantum optimizer, while retaining a classical adversarial subproblem to compute worst-case recourse and certify solution quality. We construct a constraint-preserving QUBO encoding for inequality-constrained master problems using slack variables and penalty terms, enabling general mixed-integer structures to be mapped to quantum-compatible representations. To address inexactness arising from discretization, penalty modeling, and quantum optimization, we introduce a bound-adjustment mechanism that yields valid lower and upper bounds and provides a certified stopping criterion. We show that the proposed framework generalizes classical column-and-constraint generation and retains its convergence properties when the master problem is solved exactly. Numerical experiments on two-stage robust location-transportation problems demonstrate that the proposed hybrid approach achieves solution quality comparable to classical methods while reducing the computational burden associated with solving mixed-integer master problems, highlighting the potential of hybrid quantum–-classical optimization for scalable decision-making under uncertainty. 
\end{abstract}

}
\maketitle
\section{Introduction}
\label{sec:Introduction}
Quantum computing has emerged as a promising paradigm for tackling large-scale combinatorial optimization problems that are challenging for classical techniques. In particular, a wide range of discrete optimization problems can be transformed into equivalent QUBO models \cite{Lewis2017, Slongo_2025}, which serve as a standard representation for implementation on quantum annealers and related quantum architectures. This has driven growing interest in leveraging quantum optimization to address complex decision-making problems across operations research, machine learning, and engineering systems.

Despite these advances, a key challenge remains: how to effectively integrate QUBO-based quantum optimization within structured optimization frameworks while preserving solution feasibility and optimality guarantees. Most existing approaches focus on standalone QUBO formulations or heuristic integration of quantum solvers, often lacking mechanisms for rigorous solution certification. This limitation is particularly critical in optimization applications involving decision-making under uncertainty, where guarantees on solution quality are essential.

Two-stage robust optimization (2SRO) provides a fundamental approach to modeling sequential decision-making in the presence of uncertainty, with wide-ranging applications in energy systems, transportation, and supply chains \cite{BenTal2004, ZengZhao2013, Dumouchelle2024, AronBrenner2025}. In this framework, first-stage decisions are made before uncertainty realization, while second-stage recourse actions adapt to the realized scenario. This modeling captures the sequential nature of many real-world decision processes and addresses the trade-off between robustness and flexibility.

However, solving two-stage RO problems remains computationally challenging due to their inherent \textit{min--max--min} structure. A widely used exact solution approach is column-and-constraint generation (CCG), which iteratively alternates between solving a restricted master problem and an adversarial subproblem that identifies worst-case realizations of uncertainty \cite{ZengZhao2013}. While CCG avoids explicit enumeration of the uncertainty set and provides strong convergence guarantees, its computational performance is often limited by the repeated solution of increasingly large mixed-integer master problems as new scenarios are generated, introducing additional recourse-variable blocks and associated constraints at each iteration \cite{ZengZhao2013, ManTsang2023}. Consequently, the master problem becomes the dominant computational bottleneck. 

This bottleneck motivates the integration of quantum optimization into decomposition-based algorithms.
Unlike classical computation, quantum optimization exploits quantum effects such as superposition and tunneling to enhance exploration of complex high-dimensional solution spaces, potentially enabling more efficient search in highly nonconvex and combinatorial optimization problems \cite{AlbashLidar2018, Gemeinhardt2023, Rudraksh2026, Volpe2025, Quinton2025}.  However, directly replacing the classical master solver with a quantum optimizer poses significant challenges, as quantum solutions are inherently approximate and may fail to satisfy the feasibility or optimality requirements of the original problem. Thus, integrating quantum optimization into decomposition-based algorithms must explicitly account for approximation errors while maintaining correctness and solution guarantees. 

To address this challenge, we develop a novel hybrid quantum-classical column-and-constraint generation (QCCG) algorithm to solve the 2SRO problems. The proposed hybrid approach preserves the decomposition structure of classical CCG. Specifically, the key idea is to reformulate the restricted master problem as a QUBO model and solve it approximately using a quantum optimizer, while retaining a classical subproblem for exact worst-case evaluation and solution certification. In this framework, quantum optimization efficiently generates high-quality candidate first-stage decisions, while classical optimization ensures feasibility and provides rigorous bounds. To the best of our knowledge, \textit{this is the first work to incorporate QUBO-based quantum optimization into a column-and-constraint generation framework while preserving theoretical guarantees under inexact master solutions}. Our main contributions are as follows:

\begin{itemize}
    \item We develop a hybrid QCCG algorithm for solving 2SRO problems. The proposed algorithm preserves the decomposition structure of classical CCG by solving the subproblem exactly using classical optimization, while leveraging quantum optimization to handle the combinatorial master problem.
    \item We propose a constraint-preserving QUBO encoding of the restricted master problem, enabling mixed-integer structures with scenario-dependent recourse constraints to be mapped into a quantum-compatible representation.
    \item We introduce a bound-adjustment mechanism that explicitly accounts for discretization error, penalty-model mismatch, and quantum optimization error. This yields valid lower and upper bounds and enables a certified stopping criterion despite inexact master solutions.
    \item We show that when the RMP is solved exactly, the proposed QCCG reduces to the classical CCG method and thus retains the finite-convergence property of the classical algorithm under standard assumptions. When the RMP is solved approximately at iteration $k$ within a tolerance $\delta_k$, the corrected lower bound remains valid, and the algorithm still returns an $\varepsilon$-optimal solution once the optimality gap is sufficiently small.
    \item Numerical results demonstrate that the proposed hybrid framework achieves solution quality of classical methods while significantly reducing the computational burden of solving mixed-integer master problems.
\end{itemize}
The rest of the paper is organized as follows. Section~\ref{sec:related} provides an overview of related work. Section~\ref{sec:2SRO} presents the formulation of the 2SRO problem and reviews the classical CCG method. Section~\ref{sec:QCCCG} introduces the proposed hybrid  CCG framework and the QUBO formulation of the master problem as well as convergence properties of the proposed method. Section~\ref{sec:Result} presents numerical evaluations. Finally, the paper concludes in Section \ref{sec:Conclusion}.
\section{Related Work} 
\label{sec:related}
\textbf{Quantum Optimization:} Recent advances in quantum computing have attracted significant attention as a potential approach for tackling large-scale combinatorial optimization problems \cite{AlbashLidar2018, Gemeinhardt2023, Rudraksh2026}. Specifically, many such problems can be formulated as QUBO problems or equivalent Ising models, which provide the standard interface for quantum hardware \cite{Lucas2014, Glover2019}. 
This capability has motivated growing research on applying quantum optimization techniques to combinatorial problems arising in operations research and machine learning. Among these approaches, quantum annealing is designed to solve Ising problems by evolving a quantum system toward the ground state of a problem Hamiltonian \cite{AlbashLidar2018, Kim2025, Rudraksh2026}. In parallel, gate-based variational quantum algorithms, such as the quantum approximate optimization algorithm, employ parameterized quantum circuits optimized through classical feedback to address combinatorial optimization tasks \cite{Farhi2014, Cerezo2021, Blekos2024}.  While these methods show promise, \textit{most existing work focuses on standalone QUBO formulations and does not consider how to embed them within structured optimization frameworks with theoretical guarantees.}

\noindent \textbf{Hybrid Quantum-Classical Optimization}:
To bridge this gap, hybrid quantum-classical optimization frameworks have been proposed, combining quantum solvers with classical algorithmic structures. 
One prominent example is hybrid quantum Benders decomposition, where the master problem is reformulated as a QUBO and solved using quantum hardware while the subproblem remains classical \cite{ZhaoFan2022, Duong2026}. Subsequent work has explored similar ideas across different architectures and applications \cite{Naghmouchi2024, Leenders2024}. However, \textit{existing hybrid approaches primarily focus on deterministic optimization and do not explicitly address decision-making under uncertainty.}

\noindent \textbf{Robust Optimization and CCG:} 2SRO is a widely used framework for decision-making under uncertainty.
CCG \cite{ZengZhao2013} is a standard approach for solving 2SRO and has been successfully applied in various domains, such as power systems \cite{ccgpower, ccgpower2}, computer networking \cite{ccgnet1,ccgnet2}, and transportation \cite{Faiz2022}. 
Despite its effectiveness, CCG can become computationally expensive as the master problem grows with the number of scenarios. This raises the question of whether quantum optimization can help alleviate this bottleneck.

\noindent \textbf{This work:} The existing hybrid methods lack mechanisms for handling inexact quantum solutions while preserving theoretical guarantees. Furthermore, the integration of quantum optimization into CCG for 2SRO remains largely unexplored.
This paper addresses this gap by developing a hybrid QCCG framework with certified bounds under inexact quantum master solves, providing a principled approach for integrating QUBO-based quantum optimization into decomposition-based robust optimization.

\section{Two-stage Robust Optimization}
\label{sec:2SRO}
\subsection{Problem formulation}
\label{subsec:reformulation}
We consider the standard \emph{min-max-min} formulation of the two-stage adaptive robust optimization (2SRO) problem: 
\begin{subequations}
\label{eq:2SRO_original}
\begin{align}
    \min_{\mathbf x\in \mathcal X} 
    & \Big\{ \mathbf c^\top\mathbf x + \max_{\boldsymbol\xi \in \mathcal U} \; \min_{\mathbf y \in \Omega (\mathbf x, \boldsymbol \xi)} \; \mathbf d^\top\mathbf y \Big\},
\end{align}
where the second-stage feasible region is given by:
\begin{align}
    \Omega \left(\mathbf{x}, \boldsymbol{\xi} \right) = \left\{ \mathbf{y} \in \mathbb R^m_{+}: \mathbf{C} \mathbf{x} + \mathbf{D} \mathbf{y} +  \mathbf{E} \boldsymbol{\xi} \geq \mathbf{h} \right\}.
\end{align}
\end{subequations}
Here, the first-stage decision $\mathbf x$ is selected before the uncertainty $\boldsymbol \xi$ is realized, whereas the second-stage \emph{recourse} decision $\mathbf y$ is determined after observing $\boldsymbol \xi$ to minimize the associated recourse cost. The set $\mathcal U$ represents the uncertainty set containing all possible realizations of the uncertain parameters.
The vectors $\mathbf c \in \mathbb R^n$, $\mathbf d \in \mathbb R^m$, $\mathbf h \in \mathbb R^r$, along with the matrices $\mathbf C \in \mathbb R^{r\times n}$, $\mathbf D \in \mathbb R^{r\times m}$, $\mathbf E \in \mathbb R^{r\times \ell}$ are given problem data. 
The problem~\eqref{eq:2SRO_original} seeks a solution that is robust against the worst-case realization of uncertainty. To ensure tractability, we impose the following widely-used assumptions throughout the paper \cite{ZengZhao2013}. 

\begin{assumption}[First-stage feasible region]
\label{ass:feasible_x}
    The first-stage feasible set $\mathcal X$ is a nonempty, bounded polyhedron with binary restrictions, e.g., $\mathcal X = \{ \mathbf{x} \in \{0,1\}^n : \mathbf{A} \mathbf{x} \ge \mathbf{b}; \mathbf A \in \mathbb R^{p \times n}, \mathbf b \in \mathbb R^p \}.$
\end{assumption}

\begin{assumption}[Uncertainty set]
\label{ass:uncertainty}
    $\mathcal U$ is assumed to be nonempty, compact and polyhedral, in particular $\mathcal U = \{\boldsymbol \xi: \mathbf F \boldsymbol \xi \le \mathbf g\}$ (hence, the worst-case over $\mathcal U$ is attained at an extreme point). A finite discrete set is a special case. 
\end{assumption}

\begin{assumption}[Relatively complete recourse]
\label{ass:complete_recourse}
    For every $(\mathbf x, \boldsymbol \xi) \in \mathcal X \times \mathcal U$ the second-stage feasibility set $\Omega(\mathbf x, \boldsymbol \xi)$ is nonempty (i.e., $\Omega(\mathbf x, \boldsymbol \xi) \neq \emptyset$). 
\end{assumption}

\begin{assumption}[Strong duality]
\label{ass:strong_duality}
    For every $(\mathbf x, \boldsymbol\xi) \in \mathcal X \times \mathcal U$ the inner minimization problem (recourse problem) is finite (i.e., $\min_{\mathbf{y} \in \Omega \left(\mathbf{x}, \boldsymbol{\xi} \right)} \mathbf{d}^\top \mathbf{y} > -\infty$) and satisfies strong duality.
\end{assumption}

Under these assumptions, it is convenient to introduce a recourse value function. For any fixed first-stage decision $\mathbf{x} \in \mathcal X$ and uncertainty realization $\boldsymbol \xi \in \mathcal U$, the second-stage (recourse) value function can be expressed as: 
\begin{align}
\label{eq:Q_def}
    & \mathcal Q(\mathbf x, \boldsymbol \xi) := \min_{\mathbf y \in \mathbb R^m_+} \Big\{\mathbf d^\top \mathbf y: \mathbf D \mathbf y \ge \mathbf h - \mathbf C \mathbf x - \mathbf E \boldsymbol \xi \Big\}.
\end{align}
The worst-case recourse function associated with a first-stage decision $\mathbf{x}$ is then defined as follows:
\begin{align} 
\label{eq:Phi_robust_value} 
    & \Phi(\mathbf{x}) = \max_{\boldsymbol{\xi} \in \mathcal U} \mathcal{Q}(\mathbf{x},\boldsymbol{\xi}). 
\end{align}
Hence, the original 2SRO problem~\eqref{eq:2SRO_original} can then be written equivalently via the worst-case recourse function as follows:
\begin{align}
\label{eq:2SRO_reform}
    & z^\star  = \min_{\mathbf{x} \in \mathcal X}  \mathbf{c}^\top \mathbf{x} + \Phi(\mathbf{x}).
\end{align}

The main challenge in solving ~\eqref{eq:2SRO_reform} lies in evaluating $\Phi(\mathbf x)$, which requires solving a nested worst-case optimization problem over the uncertainty set and the second-stage recourse problem. Column-and-constraint generation (CCG) offers an effective approach to address this challenge \cite{ZengZhao2013}. The main idea of CCG is to avoid enumerating all uncertainty realizations in advance and instead, it iteratively constructs a RMP that includes only those scenarios that are essential to the worst-case structure of the current solution \cite{ZengZhao2013}.

\begin{remark} \textit{ 
    Under Assumptions~\ref{ass:uncertainty}--\ref{ass:strong_duality}, the recourse value function $\mathcal{Q}(\mathbf{x},\boldsymbol{\xi})$ is convex in $\mathbf{x}$ for each fixed $\boldsymbol\xi$, and the worst-case value $\Phi(\mathbf{x}) = \max_{\boldsymbol{\xi} \in \mathcal U} \mathcal{Q}(\mathbf{x}, \boldsymbol{\xi})$ is also convex as it is the pointwise maximum of affine functions in $\mathbf{x}$. This convexity is essential to ensure the validity of the CCG}. 
\end{remark}

\subsection{Column-and-Constraint Generation}
\label{subsec:ccg}
To handle the worst-case recource function $\Phi(\mathbf x)$, we introduce an auxiliary variable $\eta$ that represents an upper bound on the worst-case recourse cost. This leads to an epigraph equivalent reformulation of~\eqref{eq:2SRO_reform} as follows:
\begin{align}
\label{eq:epigraph_model}
    & \min_{\mathbf x,\eta} \left\{ \mathbf c^\top \mathbf x + \eta : \mathbf x \in \mathcal X, \eta \ge \Phi(\mathbf x) \right\}.
\end{align}
The epigraph constraint $\eta \ge \Phi(\mathbf{x})$ ensures that the auxiliary variable $\eta$ upper-bounds the worst-case recourse value. At optimality, the bound is attained, i.e., $\eta = \Phi(\mathbf{x})$, implying that the epigraph formulation is equivalent to the original robust problem.
From~\eqref{eq:Q_def}, the epigraph constraint $\eta \ge \Phi(\mathbf x)$ can be expanded as follows:
\begin{align}
\label{eq:eta_expand}
    & \eta \ge \mathcal Q (\mathbf x, \boldsymbol \xi), \;\; \forall \boldsymbol \xi \in \mathcal U.
\end{align}
Thus, the worst-case condition is equivalent to requiring that $\eta$ upper-bound the recourse value associated with \emph{every} scenario in the uncertainty set. 
Additionally, $\mathcal Q(\mathbf x,\boldsymbol\xi)$ is the optimal value of the recourse problem \eqref{eq:Q_def}. For a fixed scenario $\boldsymbol\xi^i$, the inequality $\eta \ge \mathcal Q (\mathbf x,\boldsymbol\xi^i)$ holds if and only if there exists a recourse $\mathbf y^i$ such that:
\begin{align}
\label{eq:y_cond}
    & \Big\{ \mathbf y^i \in \mathbb R^m_+, \; \mathbf D \mathbf y^i \ge \mathbf h - \mathbf C \mathbf x - \mathbf E \, \boldsymbol\xi^i, \; \eta \ge \mathbf d^\top \mathbf y^i \Big\}.
\end{align}
Indeed, if there exists a feasible recourse decision whose cost does not exceed $\eta$ for a given scenario, then $\eta$ must be at least as large as the minimum feasible recourse cost associated with that scenario. Instead of approximating the recourse value function through dual cuts, the algorithm enforces scenario-specific feasibility and cost constraints directly by introducing recourse variables and the corresponding constraints for each generated scenario.

\subsubsection{Restricted master problem}
Under Assumption~\ref{ass:uncertainty}, we assume that the uncertainty set is finite, (i.e., $\mathcal U = \{\boldsymbol \xi^1, \boldsymbol \xi^2, \dots, \boldsymbol \xi^s \}$). 
A deterministic equivalent formulation can therefore be obtained by introducing one recourse vector $\mathbf y^i$ for each scenario $\boldsymbol\xi^i \in \mathcal U$, together with the corresponding feasibility and epigraph constraints in~\eqref{eq:y_cond}. Consequently, $\eta$ upper-bounds the recourse cost for every scenario, and the objective $\mathbf{c}^\top \mathbf{x} + \eta$ represents exactly the first-stage cost together with the worst-case recourse value.

However, in general, the uncertainty set $\mathcal U$ may be too large or even infinite, making it computationally impractical to include all scenarios explicitly. The core idea of CCG is thus to work instead with a restricted master problem (RMP) constructed only from a subset of generated scenarios. Specifically, at iteration $k$, we define $\mathcal U_k := \{\boldsymbol\xi^1, \boldsymbol\xi^2, \dots, \boldsymbol\xi^k\} \subseteq \mathcal U$. The restricted master is then constructed using only the scenarios in $\mathcal{U}_k$. In particular, for each scenario $\boldsymbol\xi^i \in \mathcal U_k$, a separate recourse block $\mathbf y^i$ is introduced to enforce both recourse feasibility and the corresponding epigraph constraint.
New scenarios are added only when the subproblem identifies a worst--case realization that is not yet represented in the current master. In this way, CCG avoids solving the full deterministic equivalent problem directly, while progressively improving the master until robustness is achieved.
Accordingly, the RMP at iteration $k$ can be formulated as:
\begin{subequations}
\label{eq:RMP}
\begin{align}
    \min_{\mathbf{x},\eta,\{\mathbf{y}^i\}_{i=1}^k} & \mathbf{c}^\top \mathbf{x} + \eta \label{eq:RMP_obj} \\
    \text{s.t.} \quad 
    & \mathbf{x} \in \mathcal X, \label{eq:RMP_1}\\
    & \mathbf{y}^i \ge \mathbf 0, \; i=1,\dots,k, \\
    & \mathbf{D}\mathbf{y}^i \;\ge\; \mathbf{h} - \mathbf{C}\mathbf{x} - \mathbf{E} \, \boldsymbol{\xi}^i, \; i=1,\dots,k, \label{eq:RMP_2}\\
    & \eta \;\ge\; \mathbf{d}^\top \mathbf{y}^i, \; i=1,\dots,k. \label{eq:RMP_3}
\end{align}
\end{subequations}
Let $(\mathbf{x}^k,\eta^k,\{\mathbf{y}^{i,k}\}_{i=1}^k)$ denote an optimal solution to \eqref{eq:RMP}, and define the lower bound as $LB_k := \mathbf{c}^\top \mathbf{x}^k + \eta^k.$
Because the RMP enforces recourse constraints only over the finite scenario set $\mathcal U^k \subseteq U$, it underestimates the full worst-case recourse function $\Phi(\mathbf{x})$, and hence $LB_k \le z^\star$.
Although we considers a finite uncertainty set, the same scenario-generation approach applies when $\mathcal{U}$ is a compact polyhedron. Since the subproblem is a LP in $\boldsymbol{\xi}$, its optimal solution is attained at an extreme point of $\mathcal{U}$. Therefore, CCG iteratively identifies and incorporates extreme scenarios that are worst-case for the current first-stage decision.

\subsubsection{Subproblem}
Given the current first-stage decision $\mathbf x^k$, the subproblem evaluates its true worst-case recourse value as:
\begin{align}
\label{eq:SP_ccg}
   & \Phi(\mathbf{x}^k) \!= \max_{\boldsymbol{\xi} \in \mathcal U} \min_{\mathbf{y} \in \mathbb R^m_+} \! \left\{ \mathbf{d}^\top \mathbf{y} : \mathbf{D}\mathbf{y} \ge \mathbf{h} - \mathbf{C}\mathbf{x}^k - \mathbf{E}\boldsymbol{\xi} \right\}. \!
\end{align}
This problem is referred to as the CCG \emph{subproblem} and it searches for the uncertainty realization that maximizes the recourse cost of the current first-stage decision. Once $\Phi(\mathbf x^k)$ is obtained, we define the corresponding upper bound as:
\begin{align}
\label{eq:UB_k}
    & UB_k := \mathbf c^\top \mathbf x^k + \Phi(\mathbf x^k).
\end{align}
Since $\mathbf x^k$ is feasible for the original problem and $\Phi(\mathbf x^k)$ is its true worst-case recourse value, $UB_k$ is an upper bound on the optimal robust objective value, i.e.,  $z^\star \le UB_k.$ 

By Assumption~\ref{ass:complete_recourse}, the second-stage is feasible for every $(\mathbf x,\boldsymbol\xi) \in \mathcal X \times \mathcal U$. Hence, the subproblem returns a finite worst-case value together with at least one worst-case scenario $\boldsymbol\xi^{k+1} \in \mathcal U$. If the current scenario set is not yet sufficient to represent the true worst-case recourse function at $\mathbf x^k$, the new scenario $\boldsymbol\xi^{k+1}$ is added to the RMP by introducing a new recourse block $\mathbf y^{k+1}$ and the following constraints:
\begin{subequations}
\begin{align}
    \mathbf{D}\mathbf{y}^{k+1} & \ge \mathbf{h} - \mathbf{C}\mathbf{x} - \mathbf{E}\boldsymbol{\xi}^{k+1}, \label{eq:new_feas_ccg} \\ 
    \eta &\ge \mathbf{d}^\top \mathbf{y}^{k+1}. \label{eq:new_epi_ccg}
\end{align}
\end{subequations}
Thus, each iteration expands the RMP by adding both a \emph{block of recourse variables} and its associated \emph{constraints}.

The algorithm terminates when the lower and upper bounds are sufficiently close, which means that the optimality gap satisfies the following inequality:
\begin{align}
\label{eq:gap_stop}
    UB_k - LB_k \le \varepsilon,
\end{align}
for a predefined tolerance $\varepsilon \ge 0$. At that point, the current first-stage solution is certified to be optimal within tolerance $\varepsilon$. We summarize the procedure in Algorithm~\ref{alg:ccg}.

\begin{lemma}[Monotonicity of bounds]
\label{lem:monotone}
    Let $LB_k$ and $UB_k$ denote the lower and upper bounds obtained at iteration $k$ in Algorithm~\ref{alg:ccg}. Then the sequence $\{LB_k\}$ is nondecreasing, and the sequence $\{\min_{j \le k} UB_j\}$ is nonincreasing.
\end{lemma}
\noindent \textit{Proof:} See \textit{Appendix} \ref{pro:lem_monotone}.

\begin{theorem}[Finite convergence of CCG]
\label{thm:finite_conv}
    Suppose Assumptions~\ref{ass:uncertainty}--\ref{ass:strong_duality} hold, and Algorithm~\ref{alg:ccg} is solved exactly at every iteration. Then the algorithm terminates in finitely many iterations. If $\varepsilon = 0$, the solution is globally optimal.
\end{theorem}
\noindent \textit{Proof:} See \textit{Appendix} \ref{pro:thm_finite_conv}.

\begin{algorithm}[t]
\caption{Column-and-Constraint Generation Algorithm}
\label{alg:ccg}
\begin{algorithmic}[1]
\STATE \textbf{Input:} tolerance $\varepsilon \ge 0$. Initialize $\mathcal U^1 = \{\boldsymbol\xi^1\} \subset \mathcal U$ with any feasible scenario.
\FOR{$k=1,2,\dots$}
  \STATE Solve the RMP \eqref{eq:RMP} over scenarios $\{\boldsymbol\xi^i\}_{i=1}^k$ to obtain $(\mathbf x^k,\eta^k,\{\mathbf y^{i,k}\}_{i=1}^k)$.
  \STATE Set $LB_k = \mathbf c^\top \mathbf x^k + \eta^k$.
  \STATE Solve the problem \eqref{eq:SP_ccg} at $\mathbf x^k$ to obtain $\Phi(\mathbf x^k)$ and a worst-case scenario $\boldsymbol\xi^{k+1}$.
  \STATE Set $UB_k = \mathbf c^\top \mathbf x^k + \Phi(\mathbf x^k)$.
  \IF{$UB_k - LB_k \le \varepsilon$}
    \STATE \textbf{Terminate} and return $\mathbf x^k$.
  \ENDIF
  \STATE Update $\mathcal U^{k+1} = \mathcal U^k \cup \{\boldsymbol\xi^{k+1}\}$.
\ENDFOR
\end{algorithmic}
\end{algorithm}
 
\section{Hybrid Quantum--Classical Algorithm for Two-Stage Adaptive Robust Optimization}
\label{sec:QCCCG}
The main computational bottleneck of the classical CCG method lies in repeatedly solving the RMP. To address this, we propose a hybrid quantum--classical CCG framework in which the subproblem remains classical and is solved exactly, while the RMP is solved approximately using a quantum optimizer after binary encoding and QUBO reformulation. A key difficulty is that an approximate quantum master solve does not directly yield a valid lower bound for the original robust problem. Thus, unlike classical CCG, the proposed hybrid method requires a certification mechanism to quantify the inexactness of the quantum master solution.

\subsection{Quantum-assisted RMP}
\label{subsec:RMP2QUBO}
To enable the use of quantum optimization for solving the RMP, it is necessary to transform it into a QUBO representation, referred to as \emph{RMP-2-QUBO}. A critical requirement in this transformation is to preserve the global minimizers of the original problem. To achieve this, all inequality constraints are first converted into equivalent equality constraints through the introduction of nonnegative slack variables.

For each uncertainty scenario $\boldsymbol\xi^i\in\mathcal U_k$ and its corresponding recourse vector $\mathbf y^i\in\mathbb R^m$, we define the stacked recourse and uncertainty vectors as follows:
\begin{subequations}
\begin{align}
\label{eq:qc_Yk}
    \mathbf Y_k &:= \big[(\mathbf y^1)^\top , (\mathbf y^2)^\top, \cdots,  (\mathbf y^k)^\top\big]^\top \in \mathbb R^{km}, \\
\label{eq:qc_Xik}
    \boldsymbol\Xi_k &:= \big[(\boldsymbol\xi^1)^\top, (\boldsymbol\xi^2)^\top\cdots, (\boldsymbol\xi^k)^\top\big]^\top \in \mathbb R^{k\ell}. 
\end{align}
\end{subequations}
Let $\mathbf{e}_k \in \mathbb{R}^k$ denote the all-ones vector, then, using the Kronecker product $\otimes$, we introduce the block matrices as: 
\begin{subequations}
\begin{align}
\label{eq:qc_block_defs}
    \mathbf D_k &:= I_k \otimes \mathbf D \in \mathbb R^{kr\times km},\\
    \mathbf C_k &:= \mathbf e_k \otimes \mathbf C \in \mathbb R^{kr\times n},\\
    \mathbf E_k &:= I_k \otimes \mathbf E \in \mathbb R^{kr\times k\ell},\\
    \mathbf h_k &:= \mathbf e_k \otimes \mathbf h \in \mathbb R^{kr},\\
    \mathbf d_k^\top &:= I_k \otimes \mathbf d^\top \in \mathbb R^{k\times km},
\end{align}
\end{subequations}

\subsubsection{Slack reformulation}
To obtain a QUBO, we rewrite all inequality constraints~\eqref{eq:RMP_1}, \eqref{eq:RMP_2}, \eqref{eq:RMP_3} as equalities by introducing the following nonnegative slack variables $\boldsymbol\sigma_k^{x} \in \mathbb R_+^{p}, \;
    \boldsymbol\sigma_k^{r} \in \mathbb R_+^{kr}, \; 
    \boldsymbol\sigma_k^{\eta} \in \mathbb R_+^{k}, $
which correspond to the first-stage, recourse, and epigraph constraints, respectively.
Accordingly, the RMP~\eqref{eq:RMP} can be equivalently rewritten as follows:
\begin{subequations}
\label{eq:qc_RMP_slack}
\begin{align}
\label{eq:qc_RMP_obj}
    \min_{\mathbf x, \eta, \mathbf Y_k, \left\{ \boldsymbol \sigma_k \right\}} 
    & \mathbf c^\top \mathbf x + \eta \\
    \text{s.t.} \quad
    & \mathbf A \mathbf x - \mathbf b - \boldsymbol \sigma_k^{x} = \mathbf 0, \label{eq:qc_RMP_eq1} \\
    & \mathbf D_k \mathbf Y_k + \mathbf C_k \mathbf x + \mathbf E_k \boldsymbol\Xi_k - \mathbf h_k - \boldsymbol \sigma_k^{r} = \mathbf 0, \label{eq:qc_RMP_eq2}\\
    & \eta \mathbf e_k - \mathbf d_k^\top \mathbf Y_k - \boldsymbol \sigma_k^{\eta} = \mathbf 0, \label{eq:qc_RMP_eq3}\\
    & \mathbf Y_k \ge \mathbf 0, \; \mathbf x \in \left\{0, 1\right\}^n, \\
    & \boldsymbol \sigma_k^{x} \ge \mathbf 0,\;
      \boldsymbol \sigma_k^{r} \ge \mathbf 0,\;
      \boldsymbol \sigma_k^{\eta} \ge \mathbf 0.
\end{align}
\end{subequations}

\subsubsection{Binary encoding}
To obtain a quantum-compatible formulation, all continuous variables in \eqref{eq:qc_RMP_slack} are discretized through 
binary expansions. Assume $\eta \in [\underline{\eta}, \overline{\eta}]$, each component of $\mathbf{Y}_k$ satisfies $0 \le (\mathbf{Y}_k)_j \le \overline{Y}_j$, and the slack variables are bounded, i.e., $0 \le (\boldsymbol{\sigma}_k^x)_a \le \overline{\sigma}^x_a, \,
    0 \le (\boldsymbol{\sigma}_k^r)_u \le \overline{\sigma}^r_u, \,
    0 \le (\boldsymbol{\sigma}_k^\eta)_i \le \overline{\sigma}^\eta_i$. 
Let $\Delta_\eta > 0$, $\Delta^Y_j > 0$, $\Delta^x_a > 0$, $\Delta^r_u > 0$, and $\Delta^\eta_i > 0$ denote discretization steps, and let $L_\eta$, $L^Y_j$, $L^x_a$, $L^r_u$, and $L^\eta_i$ be the corresponding bit-lengths chosen to cover the prescribed ranges.
The epigraph variable $\eta$ is then encoded as follows:
\begin{align}
\label{eq:eta_binary}
    & \eta = \underline{\eta} + \Delta_\eta \sum_{\ell=0}^{L_\eta-1} 2^\ell s_\ell, \; s_\ell \in \{0,1\},
\end{align}
where the bit-length $L_\eta$ is selected such that $\underline{\eta} + \Delta_\eta (2^{L_\eta} - 1) \ge \overline{\eta}$, which is equivalent to: 
\begin{align}
    & L_\eta = \Big\lceil \log_2\!\Big( \frac{\overline{\eta} - \underline{\eta}}{\Delta_\eta} + 1 \Big) \Big\rceil.
\end{align}
Similarly, for each $j$, $a$, $u$, and $i$, we define the following binary encoding for recourse, epigraph and slack variables:
\begin{subequations}
\begin{align}
\label{eq:Y_binary}
    (\mathbf Y_k)_j &= \Delta^Y_j \! \sum_{\ell=0}^{L^Y_j-1} \! 2^\ell z_{j\ell}, \; (\boldsymbol\sigma_k^x)_a = \Delta^x_a \! \sum_{\ell=0}^{L^x_a-1} \! 2^\ell u^{x}_{a\ell}, \\
\label{eq:sigma_r_binary}
    (\boldsymbol\sigma_k^r)_u &= \Delta^r_u \! \sum_{\ell=0}^{L^r_u-1} \! 2^\ell u^{r}_{u\ell}, \; 
    (\boldsymbol\sigma_k^\eta)_i = \Delta^\eta_i \! \sum_{\ell=0}^{L^\eta_i-1} \! 2^\ell u^{\eta}_{i\ell}, 
\end{align}
\end{subequations}
where $z_{j\ell}, u^{x}_{a\ell}, u^{r}_{u\ell}, u^{\eta}_{i\ell} \in \{0,1\}$ are binary variables. The corresponding bit-lengths are selected to ensure that that the prescribed variable ranges are properly represented.

\subsubsection{RMP-2-QUBO}
Let $\mathbf v \in \{0,1\}^{N_k}$ denote the binary vector obtained by stacking all binary variables as:
\begin{align}
\label{eq:binary_vector_qc}
    & \mathbf v = \big[\mathbf x^\top, \mathbf s^\top, \mathbf z^\top, (\mathbf u^x)^\top, (\mathbf u^r)^\top, (\mathbf u^\eta)^\top \big]^\top \! \! \in \{0,1\}^{\! N_k}, 
\end{align}
where $\mathbf s$ contains the binary variables used to encode $\eta$, $\mathbf z$ collects the binary variables associated with $\mathbf Y_k$, and $\mathbf u^x$, $\mathbf u^r$, and $\mathbf u^\eta$ represent the binary variables associated with the first-stage, recourse, and epigraph slack vectors, respectively.
Each original variable can be expressed as an affine function of the variable vector $\mathbf v$. Hence, there exist a matrix $\mathbf T_k$ and a vector $\mathbf t_k$ such that:
\begin{subequations}
\label{eq:zk_affine_qc}
\begin{align}
\label{eq:zk_affine_qc1}
    \mathbf z_k
    & = \big[\mathbf x^\top, \, \eta, \, \mathbf Y_k^\top, \, (\boldsymbol\sigma_k^x)^\top, \, (\boldsymbol\sigma_k^r)^\top, \, (\boldsymbol\sigma_k^\eta)^\top \big]^\top \\
\label{eq:zk_affine_qc2}
    & = \mathbf T_k \, \mathbf v + \mathbf t_k.
\end{align}
\end{subequations} 
Moreover, linear equalities \eqref{eq:qc_RMP_eq1}--\eqref{eq:qc_RMP_eq3} can be written as:
\begin{align}
\label{eq:linear_eq_qc}
    \mathbf M_k \, \mathbf z_k = \mathbf p_k,
\end{align}
where, the matrix $\mathbf M_k$ and the vector $\mathbf p_k$ are defined as: 
\begin{subequations}
\begin{align}
\label{eq:Mk_qk_qc}
    & \mathbf p_k = \Big[\mathbf b^\top \;\; (\mathbf h_k - \mathbf E_k \boldsymbol\Xi_k)^\top \;\; \mathbf 0_k^\top \Big]^\top, \\
    & \mathbf M_k =
        \begin{bmatrix}
            \mathbf A & \mathbf 0 & \mathbf 0 & -I_p & \mathbf 0 & \mathbf 0\\
            \mathbf C_k & \mathbf 0 & \mathbf D_k & \mathbf 0 & -I_{kr} & \mathbf 0\\
            \mathbf 0 & \mathbf e_k & -\mathbf d_k^\top & \mathbf 0 & \mathbf 0 & -I_k
        \end{bmatrix}.
\end{align}
\end{subequations}
Substituting \eqref{eq:zk_affine_qc2} into \eqref{eq:linear_eq_qc} we can obtain the equivalent binary linear equality as:
\begin{align}
\label{eq:binary_eq_qc}
    \mathbf G_k \mathbf v = \mathbf g_k,
\end{align}
where $\mathbf G_k := \mathbf M_k \mathbf T_k$ and $\mathbf g_k := \mathbf p_k - \mathbf M_k \mathbf t_k$.
The encoded objective is affine in $\mathbf v$. Since the original objective~\eqref{eq:qc_RMP_obj} is $\mathbf c^\top\mathbf x+\eta$, there exist $\mathbf f_k\in\mathbb R^{N_k}$ and $\gamma_k\in\mathbb R$ such that:
\begin{align}
\label{eq:encoded_obj_qc}
    & \mathbf c^\top \mathbf x + \eta = \mathbf f_k^\top \mathbf v + \gamma_k.
\end{align}
The equality constraints \eqref{eq:binary_eq_qc} are incorporated into the objective via a quadratic penalty with non-negative parameter $\rho_k>0$. Then, we can obtain a binary optimization reformulation of the RMP~\eqref{eq:RMP}, with the penalized objective given by:
\begin{align}
\label{eq:qubo_penalty_qc}
    & \phi_k(\mathbf v) = \mathbf f_k^\top \mathbf v + \gamma_k + \rho_k \left\| \mathbf G_k \mathbf v - \mathbf g_k \right\|_2^2.
\end{align}
For sufficiently large penalty value $\rho_k$, minimization of~\eqref{eq:qubo_penalty_qc} enforces the equality constraints exactly, thereby recovering the original constrained master problem. By expanding the squared norm, the penalized master problem can then be written in standard QUBO matrix form as follows:
\begin{align}
\label{eq:qc_qubo_final}
    & \phi_k(\mathbf v) = \mathbf v^\top \mathbf Q_k \mathbf v + \mathbf q_k^\top \mathbf v + c_k,
\end{align}
where $\mathbf Q_k = \rho_k \mathbf G_k^\top \mathbf G_k$ is a symmetric matrix, $\mathbf q_k = \mathbf f_k - 2\rho_k \mathbf G_k^\top \mathbf g_k$ is a vector, and $c_k = \gamma_k + \rho_k \mathbf g_k^\top \mathbf g_k$ is a constant. The coefficients $(\mathbf Q_k, \mathbf q_k, c_k)$ depend on the original RMP data, the current scenario set $\mathcal U_k$, the discretization parameters, and the penalty weights. Problem~\eqref{eq:qc_qubo_final} defines the quantum master problem solved at iteration $k$, with the corresponding optimal value, given by:
\begin{align}
\label{eq:qc_phi_star}
    & \phi_k^\star = \min_{\mathbf v \in \{0,1\}^{N_k}} \phi_k(\mathbf v).
\end{align}
\textbf{Remark:} It is important to distinguish among the successive transformations used to construct the quantum master problem. First, the slack-variable reformulation in \eqref{eq:qc_RMP_slack} is exact, as it simply replaces inequality constraints with equivalent equality constraints and nonnegative slack variables without altering the original problem. Second, the binary encoding \eqref{eq:eta_binary}, \eqref{eq:Y_binary}, \eqref{eq:sigma_r_binary} is generally approximate because continuous variables are represented using a finite binary encoding. 
Third, the penalty-based QUBO reformulation in \eqref{eq:qubo_penalty_qc} introduces an additional source of approximation whose exactness depends on the choice of the penalty parameter $\rho_k$. Equivalence between the penalized binary model and the constrained encoded master is guaranteed only when the penalty parameter is sufficiently large. 
Thus, even if the QUBO in \eqref{eq:qc_qubo_final} is solved to global optimality, this does not imply that the original RMP is solved exactly, unless both the binary encoding and the penalty reformulation are exact.

\subsection{Bound Certification}
\label{subsec:certification_qc}
Let $\mathbf v^{k,\mathrm Q}$ denote a binary solution returned by the quantum optimizer for the QUBO~\eqref{eq:qc_qubo_final}. By decoding $\mathbf v^{k,\mathrm Q}$, we can then obtain a candidate solution, given by:
\begin{align}
    & \mathbf v^{k, \mathrm Q} \to \left( \bar{\mathbf x}^k, \bar\eta^k, \bar{\mathbf Y}_k = \{\bar{\mathbf y}^{i,k}\}_{i=1}^k \right),
\end{align}
where $\bar{\mathbf x}^k$ is the candidate first-stage decision, $\bar\eta^k$ is the decoded epigraph value, and $\bar{\mathbf y}^{i,k}$ is the decoded recourse vector associated with scenario $\boldsymbol\xi^i$. Since the QUBO is solved approximately and the master problem is obtained through a discretized reformulation of the original RMP, the decoded solution must be verified with respect to the original formulation before being incorporated into the CCG procedure.
Specifically, $\bar{\mathbf x}^k$ is verified for feasibility with respect to the original feasible set $\mathcal X$, and the decoded tuple is evaluated against the master constraints~\eqref{eq:RMP_2}--~\eqref{eq:RMP_3}. Only a first-stage solution that satisfies these feasibility conditions $\bar{\mathbf x}^k \in \mathcal X$ is passed to the corresponding subproblem.

Given a quantum feasible first-stage candidate $\bar{\mathbf x}^k \in \mathcal X$, we evaluate the true worst-case recourse value using the exact classical subproblem as follows:
\begin{align}
\label{eq:true_worst_case_qc}
    & \Phi(\bar{\mathbf x}^k) = \max_{\boldsymbol\xi \in \mathcal U} \mathcal Q(\bar{\mathbf x}^k,\boldsymbol\xi).
\end{align}
This evaluation provides a valid upper bound on the optimal robust value. Specifically, we update the upper bound as:
\begin{align}
\label{eq:valid_upper_bound_qc}
    & \overline{UB}_k := \min\left\{\overline{UB}_{k-1}, \; \mathbf c^\top \bar{\mathbf x}^k + \Phi(\bar{\mathbf x}^k) \right\}.
\end{align}
Since the worst-case evaluation~\eqref{eq:true_worst_case_qc} is exact, $\overline{UB}_k$ constitutes a valid upper bound on the robust optimal value $z^\star$.

\begin{assumption}[Discretization error]
\label{ass:qc_disc}
    Let $\widetilde{ v}_k^{\mathrm{RMP}}$ denote the exact optimum of the encoded master problem prior to penalty embedding and quantum optimization. Let $v_k^{\mathrm{RMP}}$ denote the exact optimal value of the current master problem. Then, there exist constants $L_k>0$ and $\Delta_k>0$ such that:
    \begin{align}
    \label{eq:qc_disc_stability}
        & |v_k^{\mathrm{RMP}} - \widetilde{v}_k^{\mathrm{RMP}}| \le L_k \Delta_k.
    \end{align}
\end{assumption}

\begin{assumption}[Exactness of the penalty reformulation]
\label{ass:qc_penalty_exact}
    There exists a sufficiently large penalty parameter $\rho_k$ such that every global minimizer of the penalized QUBO problem \eqref{eq:qc_qubo_final} corresponds to a feasible and optimal solution of the constrained encoded master problem.
\end{assumption}

\begin{assumption}[Quantum master suboptimality]
\label{ass:qc_anneal}
    For given a binary string $\mathbf v^{k,\mathrm Q}\in\{0,1\}^{N_k}$ obtained by the quantum annealer at iteration $k$, we have: 
    \begin{align}
    \label{eq:qc_anneal_subopt}
        & \phi_k(\mathbf v^{k,\mathrm Q}) \le \phi_k^\star + \varepsilon_k^{\mathrm{Q}},
    \end{align}
    where $\varepsilon_k^{\mathrm{Q}}\ge 0$ denotes the allowed annealing suboptimality.
\end{assumption}

\begin{assumption}[High-probability annealing guarantee]
\label{ass:qc_prob}
    For each $k$, there exists $\alpha_k \in (0,1)$ such that:
    \begin{align}
    \label{eq:qc_prob}
        & \mathbb P\!\left( \phi_k(v^{k,\mathrm Q}) \le \phi_k^\star + \varepsilon_k^{\mathrm{Q}} \right) \ge 1-\alpha_k.
    \end{align}
\end{assumption}

Since the annealer does not deterministically return a global minimizer of the QUBO, we instead assume, by Assumption~\ref{ass:qc_prob}, that at iteration $k$, it produces a binary solution whose objective value is within $\varepsilon_k^{\mathrm{Q}}$ of the true QUBO optimum with probability at least $1-\alpha_k$. Here, $\varepsilon_k^{\mathrm{Q}}$ measures the optimization error induced by the annealing process, while $\alpha_k$ is the corresponding failure probability. This assumption  provides a high-probability near-optimality guarantee for the QUBO solution, which can be incorporated into the overall master inexactness budget together with discretization and penalty-reformulation errors.

\begin{lemma}[Inexact quantum master]
\label{lem:qc_delta}
    Under Assumptions~\ref{ass:qc_disc}--\ref{ass:qc_anneal}, the decoded master objective satisfies the following bound:
    \begin{align}
    \label{eq:qc_delta_main}
        & \hat v_k \le v_k^{\mathrm{RMP}} + \delta_k,
    \end{align}
    where $ \delta_k = L_k \Delta_k + \varepsilon_k^{\mathrm{Q}}$, and $\hat v_k = \mathbf c^\top \bar{\mathbf x}^k + \bar\eta^k$ denote the decoded objective value returned by the quantum master solver. More generally, if the penalty reformulation is not exact, we have:
    \begin{align}
    \label{eq:qc_delta_general}
        & \delta_k \le L_k\Delta_k + \delta_k^{\mathrm{pen}} + \varepsilon_k^{\mathrm{Q}},
    \end{align}
    where $\delta_k^{\mathrm{pen}}\ge 0$ captures the residual penalty error.
\end{lemma}

\begin{proof}
\label{pro:lem_qc_delta}
    By Assumption~\ref{ass:qc_penalty_exact}, the optimal value of the penalized QUBO problem equals the optimal value $\widetilde{v}_k^{\mathrm{RMP}}$ of the constrained encoded RMP problem, i.e., $\phi_k^\star = \widetilde v_k^{\mathrm{RMP}}$. By Assumption~\ref{ass:qc_anneal},
    \begin{align}
    \label{eq:lem_phi_k}
        & \phi_k(\mathbf v^{k,\mathrm Q}) \le \widetilde v_k^{\mathrm{RMP}} + \varepsilon_k^{\mathrm{Q}}.
    \end{align}
    By Assumption~\ref{ass:qc_disc}, we have:
    \begin{align}
    \label{eq:lem_v_k}
        & \widetilde v_k^{\mathrm{RMP}} \le v_k^{\mathrm{RMP}} + L_k \Delta_k.
    \end{align}
    Combining these inequalities~\eqref{eq:lem_phi_k}--\eqref{eq:lem_v_k}, we can obtain:
    \begin{align}
        & \phi_k(\mathbf v^{k, \mathrm Q}) \le v_k^{\mathrm{RMP}} + L_k\Delta_k + \varepsilon_k^{\mathrm{Q}}.
    \end{align}
    After decoding, we can obtain \eqref{eq:qc_delta_main} with $\delta_k = L_k \Delta_k + \varepsilon_k^{\mathrm{Q}}$. If the penalty reformulation is inexact, an additional penalty error term $\delta_k^{\mathrm{pen}}$ must be included.
    Since the decoded objective $\hat v_k = \mathbf c^\top \bar{\mathbf x}^k + \bar\eta^k$ is the original objective component of the decoded solution, we have $\hat v_k \le \phi_k(\mathbf v^{k,\mathrm Q})$.
\end{proof}

Lemma~\ref{lem:qc_delta} establishes that the decoded master objective value exceeds the exact optimal value of the RMP by at most $\delta_k$, the total master-solve error, including the optimization error of the quantum solver, the discretization error induced by binary encoding, and the penalty-model mismatch introduced by the QUBO reformulation. Since the RMP considers only the scenario subset $\mathcal U_k$, it is a relaxation of the full robust optimization problem. Consequently,
\begin{align}
\label{eq:rmp_relaxation_bound_qc}
    v_k^{\mathrm{RMP}} \le z^\star.
\end{align}
Combining \eqref{eq:qc_delta_main} and \eqref{eq:rmp_relaxation_bound_qc}, define a certified lower bound as:
\begin{align}
\label{eq:inexact_lower_bound_qc}
    \overline{LB}_k^\delta
    :=
    \hat v_k - \delta_k
    =
    \mathbf c^\top \bar{\mathbf x}^k + \bar\eta^k - \delta_k,
\end{align}
which satisfies $\overline{LB}_k^\delta \le z^\star$.

\begin{lemma}[Certified bounds]
\label{lem:hybrid_bounds_qc}
    Suppose that, at iteration $k$, the decoded first-stage candidate $\bar{\mathbf x}^k$ satisfies $\bar{\mathbf x}^k \in \mathcal X$, the worst-case recourse value $\Phi(\bar{\mathbf x}^k)$ is computed exactly, and the inexactness tolerance $\delta_k$ satisfies~\eqref{eq:qc_delta_general}. Then, we have:
    \begin{align}
        & \overline{LB}_k^{\delta} \le z^\star \le \overline{UB}_k,
    \end{align}
    where $\overline{LB}_k^{\delta}$ and $\overline{UB}_k$ are given by~\eqref{eq:inexact_lower_bound_qc} and~\eqref{eq:valid_upper_bound_qc}.
\end{lemma}
\begin{proof}
    Since $\bar{\mathbf x}^k \in \mathcal X$, the value $\mathbf c^\top \bar{\mathbf x}^k + \Phi(\bar{\mathbf x}^k)$ is the exact robust objective value associated with the feasible first-stage decision $\bar{\mathbf x}^k$. By definition of the robust optimum $z^\star$, we therefore have:
    \begin{align}
        z^\star \le \mathbf c^\top \bar{\mathbf x}^k + \Phi(\bar{\mathbf x}^k),
    \end{align}
    By the update rule~\eqref{eq:valid_upper_bound_qc}, it follows that:
    \begin{align}
        z^\star \le \overline{UB}_k.
    \end{align}
    For the lower bound, inequality~\eqref{eq:qc_delta_main} implies that:
    \begin{align}
        & \hat v_k - \delta_k \le v_k^{\mathrm{RMP}}.
    \end{align}
    Since the RMP is a relaxation of the original robust problem (it considers only a subset of scenarios), its optimal value satisfies $v_k^{\mathrm{RMP}} \le z^\star$. Therefore, we have:
    \begin{align}
        & \overline{LB}_k^{\delta} = \hat v_k - \delta_k \le z^\star.
    \end{align}
    Combining the upper- and lower-bound inequalities we can obtain the resulting certified bound relation.
\end{proof}

\begin{theorem}[$\varepsilon$-optimality of hybrid CCG]
\label{thm:qc_eps}
    Assume that the decoded first-stage solutions are feasible, the subproblem is solved exactly at each iteration, and the algorithm terminates at iteration $K$ such that:
    \begin{align}
    \label{eq:qc_stop}
        & \overline{UB}_K - \overline{LB}_K^\delta \le \varepsilon,
    \end{align}
    then any solution
    \begin{align}
    \label{eq:qc_best}
        & \hat{\mathbf x} \in \arg\min_{1\le t\le K} \left\{ \mathbf c^\top \bar{\mathbf x}^t + \Phi(\bar{\mathbf x}^t) \right\}
    \end{align}
    satisfies:
    \begin{align}
    \label{eq:qc_eps_opt}
        & \mathbf c^\top \hat{\mathbf x} + \Phi(\hat{\mathbf x}) \le z^\star + \varepsilon.
    \end{align}
    Hence, $\hat{\mathbf x}$ is $\varepsilon$-optimal for the original robust problem.
\end{theorem}

\begin{proof}
    By Lemma~\ref{lem:hybrid_bounds_qc},
    \begin{align}
        & \overline{LB}_K^\delta \le z^\star \le \overline{UB}_K.
    \end{align}
    It follows that:
    \begin{align}
    \label{eq:UB_LB_K}
        & \overline{UB}_K - z^\star \le \overline{UB}_K - \overline{LB}_K^\delta \le \varepsilon.
    \end{align}
    By the definition of $\hat{\mathbf x}$ in~\eqref{eq:qc_best}, we obtain:
    \begin{align}
    \label{eq:hat_x_qc}
        & \mathbf c^\top \hat{\mathbf x} + \Phi(\hat{\mathbf x}) = \overline{UB}_K.
    \end{align}
    Substituting~\eqref{eq:hat_x_qc} into the inequality~\eqref{eq:UB_LB_K} yields \eqref{eq:qc_eps_opt}.
\end{proof}

\begin{theorem}[High-probability certificate]
\label{thm:qc_prob}
    Suppose Assumptions~\ref{ass:qc_disc} and \ref{ass:qc_prob} hold, and the penalty reformulation is exact. If the algorithm terminates at iteration $K$ with \eqref{eq:qc_stop}, then, with probability at least $1 - \sum_{k=1}^{K} \alpha_k$, the returned solution $\hat{\mathbf x}$ satisfies the following:
    \begin{align}
    \label{eq:qc_prob_cert}
        & \mathbf c^\top \hat{\mathbf x} + \Phi(\hat{\mathbf x}) \le z^\star + \varepsilon.
    \end{align}
\end{theorem}

\begin{proof}
    By a union bound, the annealing guarantee \eqref{eq:qc_prob} holds simultaneously for all iterations $1,\dots,K$ with probability at least $1-\sum_{k=1}^K \alpha_k$. On this event, Lemma~\ref{lem:qc_delta} implies the inexact master relation \eqref{eq:qc_delta_main} at every iteration. Therefore the conditions of Theorem~\ref{thm:qc_eps} hold, and the claimed probabilistic certificate follows.
\end{proof}

It is obvious that similar to the classical CCG, the hybrid algorithm terminates when the certified optimality gap satisfies the following bound:
\begin{align}
\label{eq:hybrid_stop_qc_rewrite}
    \overline{UB}_k - \overline{LB}_k^{\delta} \le \varepsilon.
\end{align}
At termination, the first-stage solution is certified to be $\varepsilon$-optimal. The validity of the upper and lower bounds is summarized in the following Lemma.

From Lemma~\ref{lem:hybrid_bounds_qc}, it is obvious that, if the stopping condition~\eqref{eq:hybrid_stop_qc_rewrite} is satisfied, then the current first-stage decision attains a robust objective value within $\varepsilon$ of the true optimum. Therefore, although the master problem is solved approximately on a quantum device, the algorithm maintains solution certification. In particular, if the lower-bound tolerance $\delta_k$ is chosen in accordance with~\eqref{eq:qc_delta_main}, the bound relation remains valid and $\varepsilon$-optimality is guaranteed.

\begin{corollary}[Exact-master special case]
\label{cor:exact_master}
Suppose that, for every iteration $k$,
\begin{enumerate}
    \item the binary encoding is exact, i.e., $L_k\Delta_k=0$;
    \item the penalty reformulation is exact, i.e., $\delta_k^{\mathrm{pen}}=0$;
    \item the quantum optimizer returns a global minimizer of the QUBO, i.e., $\varepsilon_k^{\mathrm{Q}}=0$.
\end{enumerate}
Equivalently, the master problem is solved exactly at every iteration, i.e., $\delta_k = 0$ for all $k$. Then, we have:
\begin{align}
    & \overline{LB}_k^\delta = v_k^{\mathrm{RMP}},
\end{align}
and the hybrid QCCG reduces to classical CCG.
\end{corollary}

\begin{proof}
Fix an iteration $k$. Under the stated assumptions, all approximation errors in the quantum master vanish, and therefore:
\begin{align}
    & \delta_k = L_k\Delta_k + \delta_k^{\mathrm{pen}} + \varepsilon_k^{\mathrm{Q}} = 0.
\end{align}
By the aggregate master inexactness Lemma~\ref{lem:qc_delta}, we have:
\begin{align}
    & \hat v_k \le v_k^{\mathrm{RMP}} + \delta_k = v_k^{\mathrm{RMP}}.
\end{align}
On the other hand, $\hat v_k$ is the objective value of a feasible solution of the current restricted master problem, while $v_k^{\mathrm{RMP}}$ is the minimum objective value over all feasible solutions of that problem. Hence:
\begin{align}
    & v_k^{\mathrm{RMP}} \le \hat v_k.
\end{align}
Combining the two inequalities we can obtain:
\begin{align}
    & \hat v_k = v_k^{\mathrm{RMP}}.
\end{align}
Since $\delta_k=0$, the certified lower bound becomes:
\begin{align}
    & \overline{LB}_k^\delta = \hat v_k - \delta_k = \hat v_k = v_k^{\mathrm{RMP}}.
\end{align}
Therefore, at every iteration, the lower-bound update in the hybrid method is exactly the same as in classical CCG.

Moreover, the upper-bound update is also identical to that of classical CCG, because it is obtained by solving the same exact classical subproblem at the current first-stage solution. Consequently, both the lower-bound and upper-bound updates coincide with those of classical CCG at every iteration. It follows that the hybrid and classical CCG generate the same sequence of master problems, the same bound updates, and therefore the same algorithmic procedure. Hence, the hybrid method reduces to classical CCG.
\end{proof}

\begin{theorem}[Finite convergence in the exact-master case]
\label{thm:finite_convergence}
Assume that:
\begin{enumerate}
    \item the subproblem is solved exactly at every iteration,
    \item the hypotheses of Corollary~\ref{cor:exact_master} hold,
    \item the set of scenarios that can be generated by the subproblem is finite.
\end{enumerate}
Then the hybrid QCCG terminates in finitely many iterations with an optimal solution of the original robust problem.
\end{theorem}

\begin{proof}
    Since the restricted master and the subproblem are solved exactly at every iteration, Corollary~\ref{cor:exact_master} implies that the QCCG algorithm reduces to the classical CCG method. By Theorem~\ref{thm:finite_conv}, the algorithm must terminate after finitely many iterations.
\end{proof}

The tolerance $\delta_k$ incorporates contributions from quantum suboptimality, discretization of the continuous master variables, and penalty-model mismatch. In particular, when the quantum master is implemented through quantum annealing, the error term $\delta_k$ can be decomposed as:
\begin{align}
\label{eq:tau_decomposition_qc}
    & \delta_k = \delta_k^{\mathrm{opt}} + \delta_k^{\mathrm{disc}} + \delta_k^{\mathrm{pen}},
\end{align}
where $\delta_k^{\mathrm{opt}}$ bounds optimization error in the QUBO solver, $\delta_k^{\mathrm{disc}}$ bounds encoding error due to finite binary resolution, and $\delta_k^{\mathrm{pen}}$ bounds the mismatch between the penalized QUBO and the original constrained master.

\subsection{Growth of the Encoded Master}
\label{subsec:size_growth_top}
The key distinction between hybrid quantum Benders decomposition and hybrid QCCG lies in how the master problem evolves over time. In Benders decomposition, the master decision variables remain fixed, and the problem is progressively refined through the addition of cuts. In contrast, in CCG, each newly generated scenario introduces an additional scenario-dependent recourse block along with new epigraph and recourse constraints. Consequently, the number of binary variables in the encoded master grows approximately as:
\begin{align}
\label{eq:size_growth_top}
    & \widetilde N_k \approx n + L_\eta + k \sum_{r=1}^{m} L_r + N_{\mathrm{slack}}(k),
\end{align}
where $\sum_{r=1}^{m}L_r$ is the number of bits needed for one recourse block, and $N_{\mathrm{slack}}(k)$ counts the slack bits introduced for constraint equalities. Equation \eqref{eq:size_growth_top} highlights the central trade-off of the method. Adding scenarios strengthens the master and improves its approximation of the robust problem, but it also enlarges the QUBO and increases the difficulty of the quantum solve. The effectiveness of the hybrid algorithm therefore depends on both sides of the decomposition: the classical subproblem must generate informative worst-case scenarios, while the quantum optimizer must remain effective as the encoded master grows.

\subsection{Hybrid  QCCG Algorithm}
\label{subsec:hybrid_ccg_alg}
At each iteration, the current restricted master is constructed from the generated scenario set, encoded into binary form, and transformed into the QUBO problem. A quantum optimizer is then used to approximately solve this QUBO, and the decoded solution provides a candidate first-stage decision $\bar{\mathbf x}^k$, along with a candidate epigraph value $\bar\eta^k$ and scenario-wise recourse values. The subproblem is then solved classically at $\bar{\mathbf x}^k$ to obtain the true worst-case recourse value $\Phi(\bar{\mathbf x}^k)$ and a new worst-case scenario $\boldsymbol\xi^{k+1}$. This new scenario is appended to the master, and the process repeats.

\begin{algorithm}[t]
\caption{Hybrid Quantum--Classical CCG}
\label{alg:hybrid_qcccg}
\begin{algorithmic}[1]
\STATE \textbf{Input:} tolerance $\varepsilon \ge 0$, initial scenario $\mathcal U_1=\{\boldsymbol\xi^1\}\subseteq\mathcal U$.
\STATE Initialize $\overline{UB}_0 = +\infty$ and $\bar{\mathbf x}^{\mathrm{best}} \leftarrow \varnothing$.
\FOR{$k=1,2,\dots$}
    \STATE Construct the current RMP \eqref{eq:RMP} over $\mathcal U_k$.
    \STATE Encode the RMP as the QUBO \eqref{eq:qc_qubo_final}.
    \STATE Approximately solve the QUBO to obtain $\mathbf v^{k,\mathrm Q}$.
    \STATE Decode $\mathbf v^{k,\mathrm Q}$ into $(\bar{\mathbf x}^k,\bar\eta^k,\bar{\mathbf Y}_k)$.
    \STATE \emph{Feasibility check}:
    \IF{\emph{decoded candidate is infeasible}}
        \STATE \emph{\textbf{reject} the candidate}.
        \STATE \emph{\textbf{continue} to the next master-solve attempt}.
    \ENDIF
    \STATE Solve the exact subproblem~\eqref{eq:true_worst_case_qc} at $\bar{\mathbf x}^k$ to compute $\Phi(\bar{\mathbf x}^k)$ and a worst-case scenario $\boldsymbol\xi^{k+1}$.
    \STATE Set $U_k^{\mathrm{cand}} \leftarrow \mathbf c^\top \bar{\mathbf x}^k + \Phi(\bar{\mathbf x}^k)$.
    \STATE Update $\overline{UB}_k \leftarrow \min\{\overline{UB}_{k-1},\, U_k^{\mathrm{cand}}\}$.
    \IF{$\overline{UB}_k = U_k^{\mathrm{cand}}$}
        \STATE $\bar{\mathbf x}^{\mathrm{best}} \leftarrow \bar{\mathbf x}^k$.
    \ENDIF
    \STATE Compute $\overline{LB}_k^\delta \leftarrow \mathbf c^\top \bar{\mathbf x}^k + \bar\eta^k - \delta_k$.
    \IF{$ \overline{UB}_k - \overline{LB}_k^{\delta} \le \varepsilon$}
        \STATE \textbf{terminate} and return $\bar{\mathbf x}^{\mathrm{best}}$.
    \ENDIF
    \STATE Update $\mathcal U_{k+1}=\mathcal U_k \cup \{\boldsymbol\xi^{k+1}\}$. 
\ENDFOR
\end{algorithmic}
\end{algorithm}

\begin{remark}
    \emph{In the proposed framework, the quantum optimizer serves as an inexact master solver and is not required to compute a globally optimal solution of the RMP at each iteration. Instead, its role is to generate high-quality candidate first-stage decisions that guide the search for critical worst-case scenarios. Certification remains classical, and the upper bound is obtained from the exact evaluation of the decoded candidate, while the lower bound is obtained by correcting the decoded master objective with a chosen inexactness tolerance. Hence, this design preserves the correctness and certification structure of classical CCG while using a quantum optimizer as a heuristic master solver.}
\end{remark}
 
\section{Numerical Results}
\label{sec:Result}
This section evaluates the performance of the proposed hybrid QCCG framework in terms of both solution quality and computational efficiency. To this end, we consider a two-stage adaptive robust location–transportation problem, in which a set of candidate facilities must be selected and sized to serve uncertain customer demand. First-stage decisions determine which facilities to open and their installed capacities, while second-stage decisions allocate transportation flows after demand realization. 

\subsection{Robust Location-transportation Problem}
Let $I$ denote the set of candidate facilities and $J$ the set of customers. The first-stage decision variables include $x_i \in \{0,1\}$ indicating whether facility $i \in I$ is opened, and the capacity decision variable $z_i \in \mathbb R_+$ representing the installed capacity at facility $i$. The second-stage decision variable $y_{ij} \in \mathbb R_+$ represents the amount of goods transported from facility $i$ to customer $j$ after demand is realized. Opening facility $i \in I$ incurs a fixed cost $f_i$, while installing capacity at that facility incurs a unit cost $a_i$. The transportation cost from facility $i$ to customer $j$ is denoted by $c_{ij}$, and the maximum capacity that can be installed at facility $i$ is $K_i$.
In practice, the customer's demand is uncertain before any facility is built and capacity is installed. To account for this uncertainty, the demand $\mathbf d$ is assumed to belong to an uncertainty set $\mathcal{D}$. We adopt a budgeted uncertainty set defined as follows:
\begin{align}
    & \!\!\! \mathcal{D} = \Big\{\mathbf{d} \! : \!d_j = \underline{d}_j + g_j \widetilde{d}_j, g_j \in \left[0, 1\right], \! \sum_{j} \!g_j \!\leq\! \Gamma, \forall j \Big\}, 
\end{align}
where $d_j$ corresponds to the nominal demand, $\widetilde{d}_j$ defines its maximum possible deviation, and $\Gamma$, a user-specified integer parameter, limits the uncertainty budget and thereby adjusts the level of conservatism. The resulting two-stage robust location--transportation problem can therefore be written as:
\begin{subequations}
\begin{align}
\label{eq:2SRO_LTP}
    & \!\!\!\!\min_{(\mathbf x, \mathbf z) \in \mathcal{S}} 
    \sum_{i\in I} \!f_i x_i + \sum_{i\in I} \! a_i z_i 
    + \max_{d \in \mathcal{D}} \min_{\mathbf y \in \mathcal{S}_y} 
    \sum_{i\in I}\sum_{j\in J} \!c_{ij} y_{ij},
\end{align}
where the first-stage feasible set is defined as:
\begin{align}
\label{eq:2SRO_LTP_S1}
    & \!\!\! \mathcal{S} \!=\! \left\{ (\mathbf y, \mathbf z): \! 
    z_i \le K_i y_i, y_i \in \{0,1\}, z_i \ge 0, \forall i \in I \right\},
\end{align}
and the second-stage feasible set is given by:
\begin{align}
\label{eq:2SRO_LTP_S2}
    & \!\!\!\! \mathcal{S}_y = \bigg\{\mathbf y \!:\!
    \sum_{j\in J} \! y_{ij} \le z_i, \forall i;
    \sum_{i\in I} \! y_{ij} \ge d_j, \forall j;
    y_{ij} \!\ge\! 0 \bigg\}.
\end{align}
\end{subequations}
The objective function~\eqref{eq:2SRO_LTP} minimizes the total system cost, including facility opening costs, capacity installation costs, and transportation costs. The first-stage constraints in~\eqref{eq:2SRO_LTP_S1} ensure that capacity is installed only at selected facilities and that it does not exceed the corresponding maximum limits. The second-stage feasible set~\eqref{eq:2SRO_LTP_S2} enforce that shipments from each facility do not exceed its installed capacity, while ensuring that the demand of each customer $j$ is satisfied.

\subsection{Performance Evaluation}
The proposed QCCG algorithm is compared with the classical CCG method, as well as classical and quantum-assisted Benders decomposition. In the hybrid approach, the RMP is reformulated as a QUBO model and solved using a quantum annealing solver from D-wave\footnote{\url{https://www.dwavequantum.com/}}, while the  subproblem is solved using Gurobi solver\footnote{\url{https://www.gurobi.com/}}. The experiments are conducted on instances of the location--transportation problem with varying numbers of facilities and customers.

The parameters of the problem are randomly generated to create various test instances. The basic demand $\underline d_j$ for each customer $j\in J$ is sampled uniformly from the interval $[10,500]$. The maximal deviation of demand is defined as $\widetilde d_j = \alpha \underline d_j$, where $\alpha$ is randomly drawn from the interval $[0.1,0.5]$, representing the relative level of uncertainty in demand. 
The maximum installable capacity at each facility $K_i$ is randomly generated from the interval $[200, 850]$, while ensuring that the total capacity is sufficient to satisfy the maximum possible demand. The fixed facility opening cost $f_i$ is drawn from $[100, 1000]$, while the unit capacity installation cost $a_i$ is sampled from $[10, 100]$. The transportation cost between facility $i$ and customer $j$, denoted by $c_{ij}$, is generated uniformly from $[1, 1000]$. To guarantee feasibility of the transportation decisions in the second stage, the generated capacities satisfy the condition $\sum_{i\in I}K_i \ge \max_{\mathbf d \in \mathcal D} \sum_{j\in J} d_j$. This ensures that the total available capacity is sufficient to meet the worst-case demand realization. 

For each problem size, 10 independent random instances are generated, and average performance is reported. For each instance, we solve both the classical CCG method (baseline) and the proposed QCCG method. We also compare our proposed QCCG method with Benders' decomposition \cite{ZengZhao2013}. All classical optimization problems are solved using Gurobi. The algorithms terminate when the optimality gap falls below a predefined tolerance; in our experiments, we set $\varepsilon = 10^{-5}$.
\vspace{-1cm}
\begin{figure}[h!]
     \subfigure[CCG: $(3 \times 3)$]{
        \includegraphics[width=0.225\textwidth]{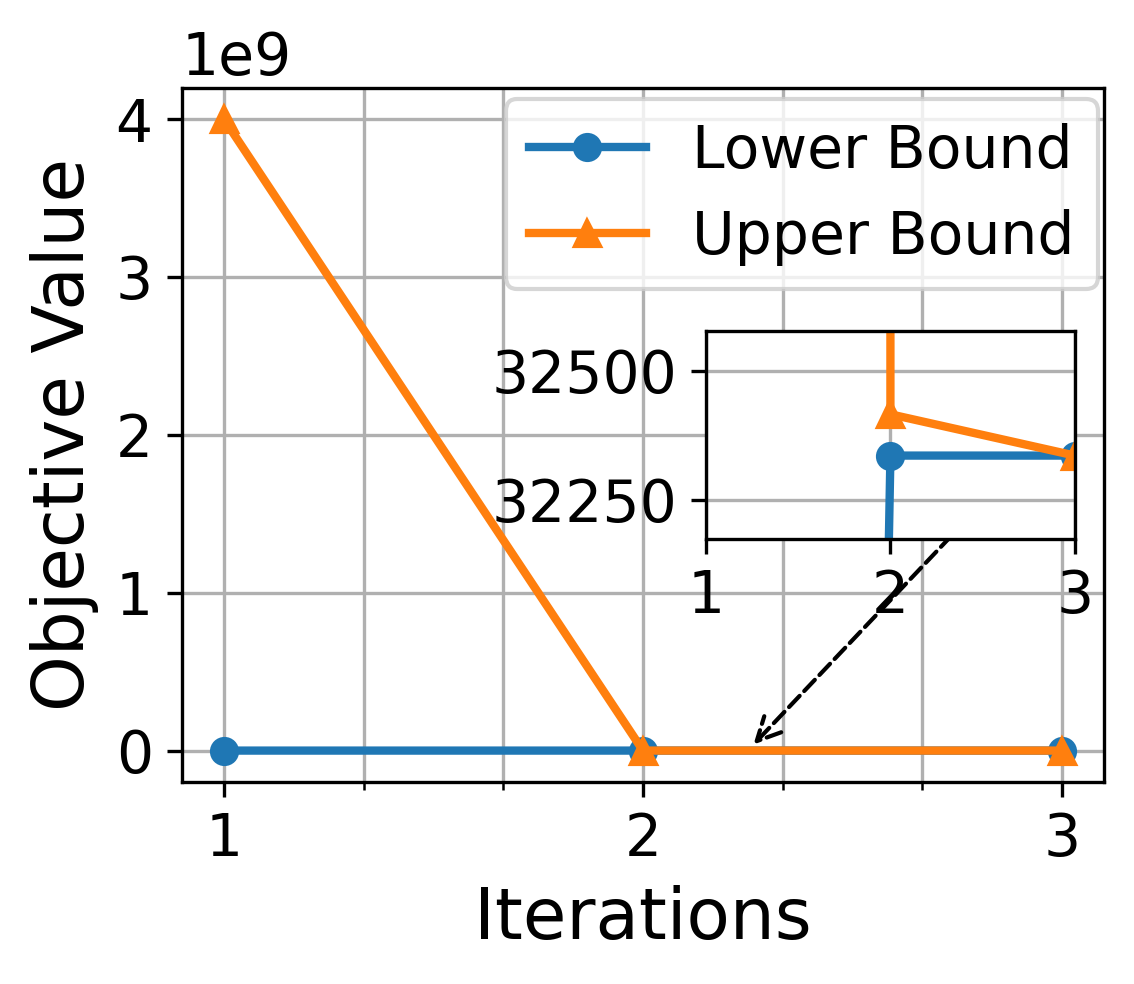}
	    \label{fig:CCG} }
    \subfigure[QCCG: $(3 \times 3)$]{
	\includegraphics[width=0.225\textwidth]{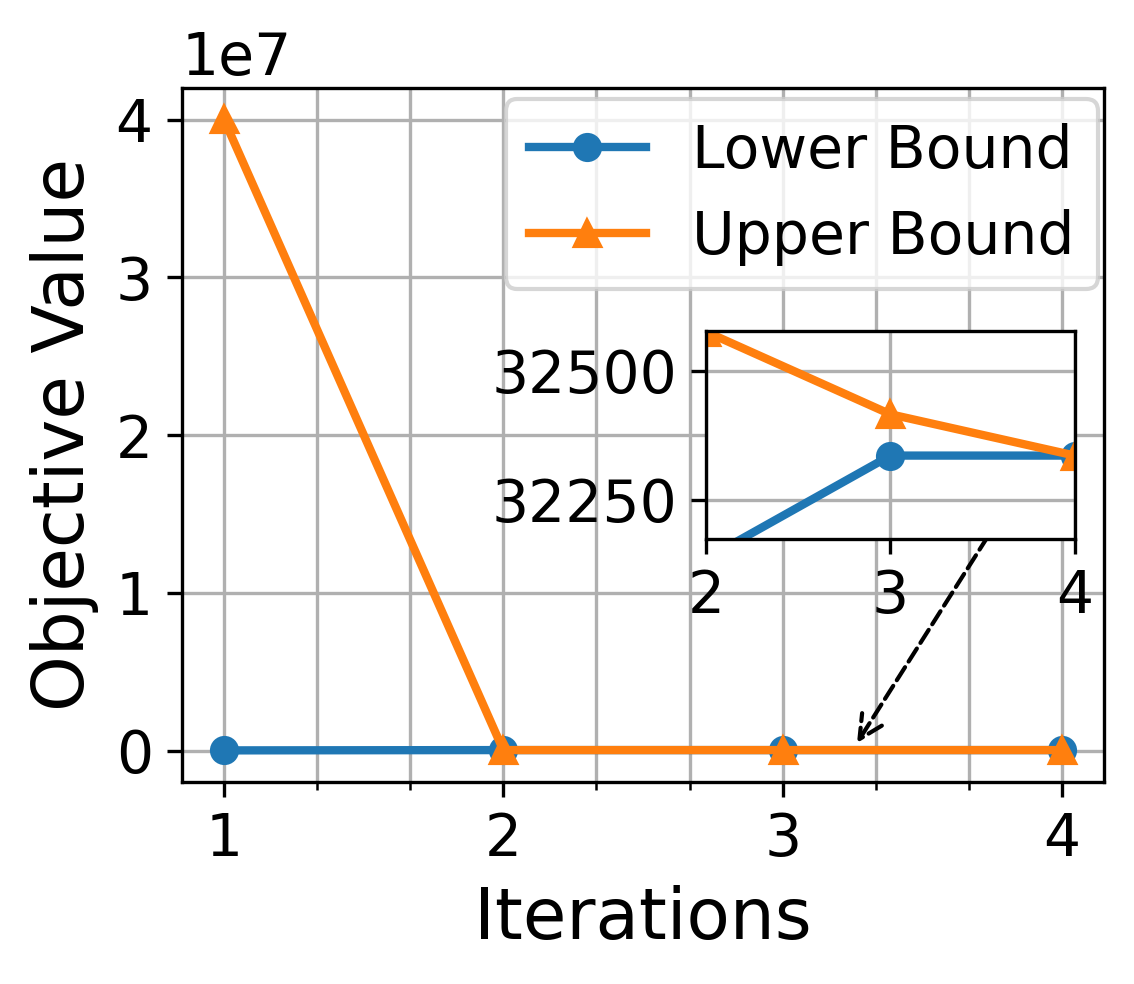}
	    \label{fig:QCCG} } 
    \subfigure[CCG: $(20 \times 20)$]{
        \includegraphics[width=0.225\textwidth]{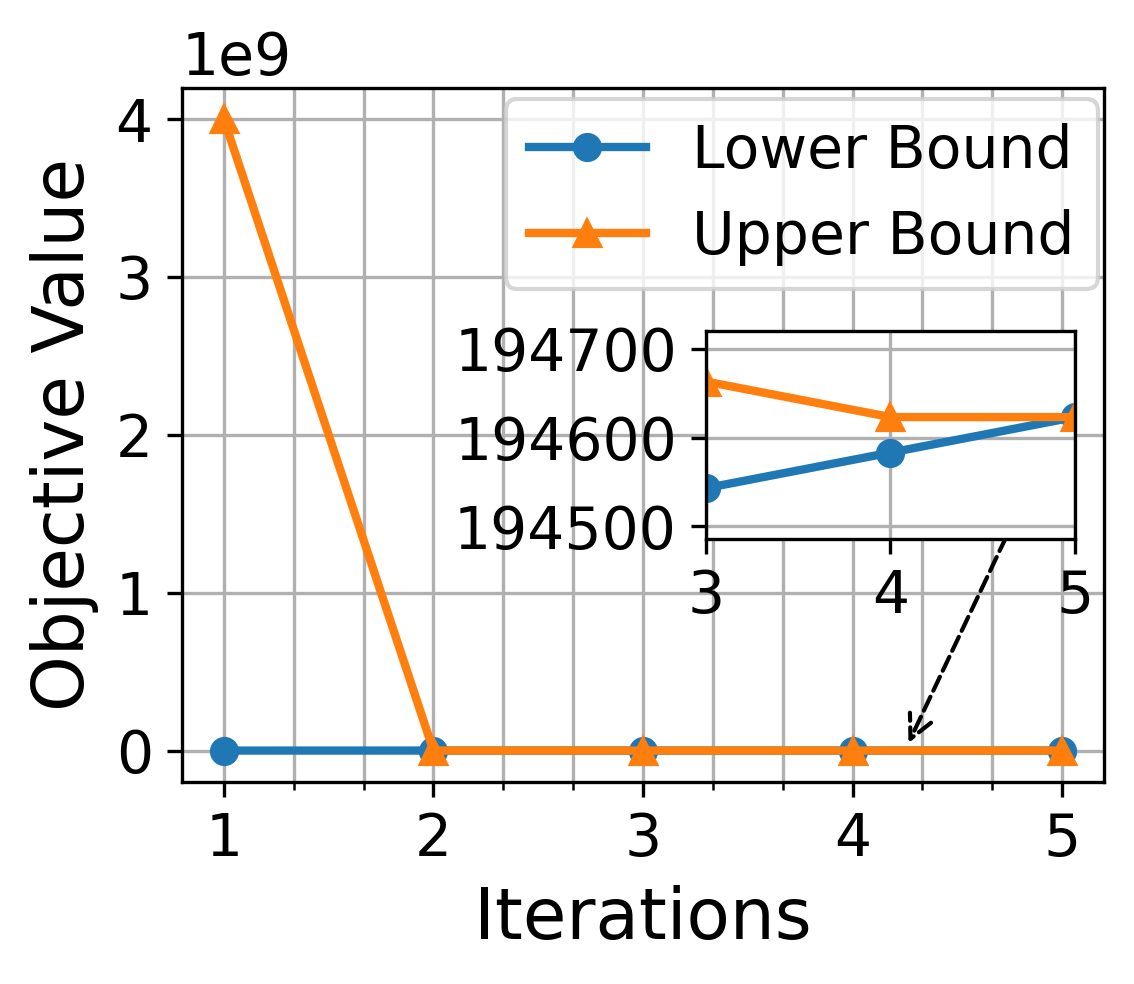}
	    \label{fig:CCG_20} }
    \subfigure[QCCG: $(20 \times 20)$]{
	\includegraphics[width=0.225\textwidth]{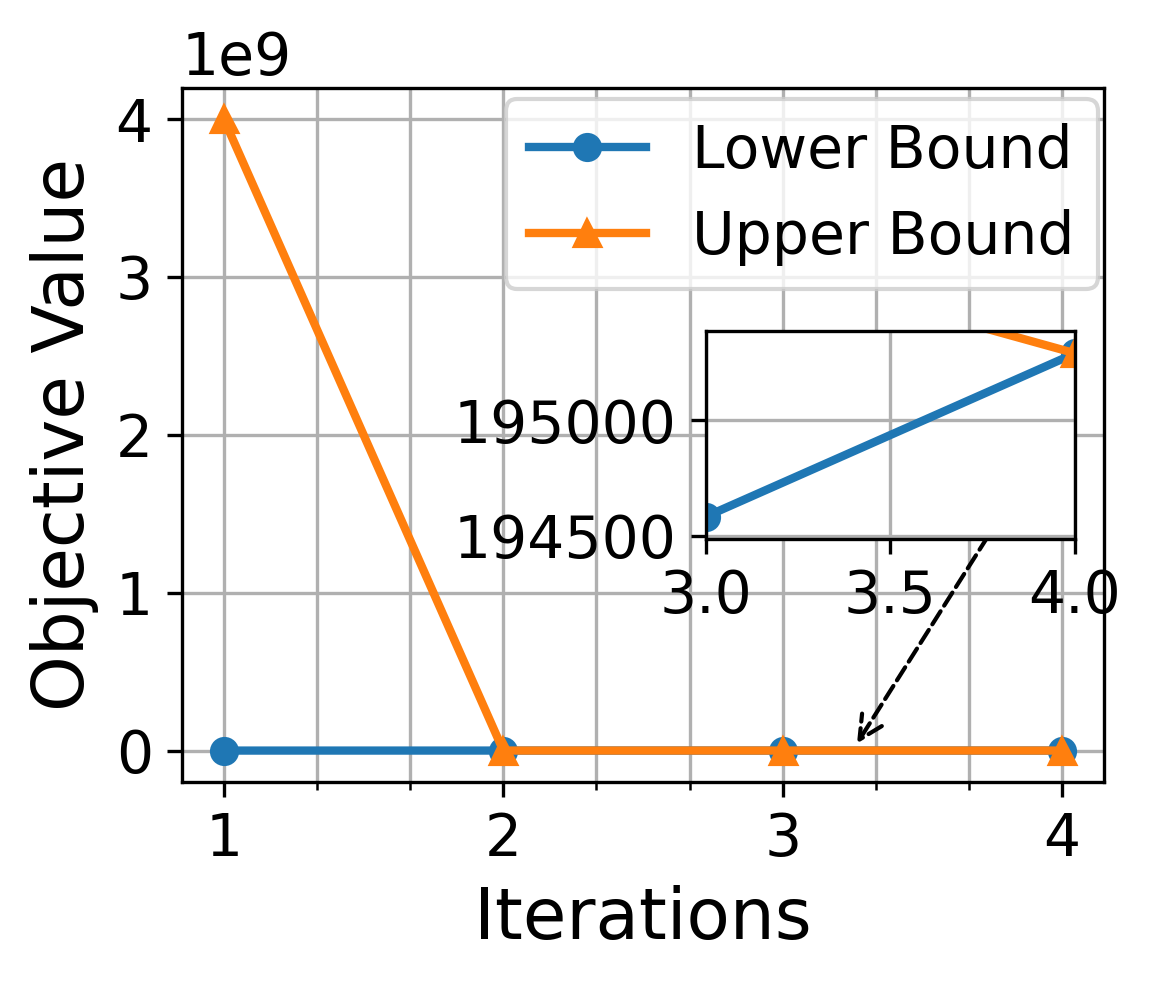}
	    \label{fig:QCCG_20} } 
	\caption{Convergence curve of CCG and QCCG.}
	\label{fig:Convergence}
    \vspace{-0.2cm}
\end{figure}

Figs.~\ref{fig:CCG}-\ref{fig:QCCG} illustrate the convergence behavior of classical CCG and the proposed QCCG algorithm for a small location--transportation instance with $3$ candidate facilities and $3$ customers. In both methods, the lower bound increases monotonically while the upper bound decreases, eventually converging to the same optimal value of approximately $32336$, which is consistent with the theoretical convergence properties of the CCG framework. It implies that our proposed QCCG algorithm is able to recover the same optimal robust solution as the classical CCG approach. 

Although both algorithms ultimately reach the same optimal value, their convergence behavior across iterations differs slightly. The classical algorithm converges in three iterations, while QCCG requires four iterations to reach the optimality gap. This difference arises from the fact that the RMP is solved using a quantum-assisted optimization procedure; the lower bound may not improve as aggressively in early iterations due to approximation errors introduced by the QUBO reformulation and quantum annealing process. Consequently, the subproblem may generate additional worst-case scenarios before convergence is achieved. 

Despite these inexact master solutions, the algorithm continues to refine the solution as additional scenarios are generated and eventually converges to the same optimal objective value, indicating that the proposed hybrid framework preserves the convergence properties of the classical CCG method.
The proposed QCCG framework also maintains the convergence behavior of the classical CCG algorithm even when the master problem is solved approximately. In particular, QCCG effectively identifies the critical uncertainty scenarios that define the worst-case cost, enabling the bounds to converge within a small number of iterations. It is also worth noting that the classical CCG method does not always converge in fewer iterations than QCCG, as illustrated in Figs.~\ref{fig:CCG_20}--\ref{fig:QCCG_20} for the 20 $\times$ 20 instance.

\begin{table}[t]
\centering
\caption{Comparison of decomposition algorithms ($3 \times 3$).}
\label{tab:method_comparison}
\begin{tabular}{lccc}
\hline
Method & Final objective & Iterations & Master solve \\
\hline
Classical Benders & 32336 & 6 & Classical \\
QC--Benders & 32336 & 6 & Quantum \\
Classical CCG & 32336 & 3 & Classical \\
Our QCCG & 32336 & 4 & Quantum \\
\hline
\end{tabular}
\end{table}

Table~\ref{tab:method_comparison} compares the convergence of different decomposition approaches. All methods reach the same final objective value of $32336$, confirming recovery of the same optimal robust solution for the location–transportation problem under the given uncertainty set. However, convergence speeds differ significantly. Classical Benders and QC--Benders each require six iterations to close the optimality gap, whereas classical CCG converges in three iterations and QCCG in four. This reflects the nature of Benders decomposition, where the master problem is iteratively refined through cut generation, often requiring multiple iterations before the critical cuts are identified. In contrast, CCG-based methods converge faster by directly identifying worst-case uncertainty realizations rather than approximating the recourse function. This highlights the structural advantage of CCG for 2SRO problems, reducing the number of iterations required for convergence. The proposed QCCG method preserves this advantage while replacing the classical master solver with a QUBO-based quantum optimization solver. Thus, QCCG maintains the convergence efficiency of classical CCG while providing a new approach for integrating quantum optimization into robust decomposition methods. To further evaluate the hybrid QCCG algorithm, we conduct experiments on scalability, the impact of the uncertainty budget, and the computational characteristics of the QUBO-based master problem, extending beyond the illustrative $3\times3$ instance to larger location–transportation problems.

Table~\ref{tab:scalability} compares the performance of classical and quantum-assisted decomposition methods on location--transportation instances of increasing size. For the small $3 \times 3$ instance, all algorithms converge to the same optimal objective value. However, the number of iterations required by the CCG-based methods is smaller than that of the Benders-based methods. This behavior is consistent with previous observations that CCG algorithms typically converge in fewer iterations due to the direct generation of significant scenarios. The results also show that the proposed QCCG framework preserves the convergence properties of the classical algorithm while enabling the use of quantum optimization to solve the combinatorial master problem. 

We evaluate solution quality by computing the optimality gap and the quantum error. The optimality gap is defined as: 
\begin{align}
    & \mathrm{Gap} = \frac{|UB-LB|}{|UB|},
\end{align}
which measures the distance between the lower and upper bounds at termination. The quantum error is defined as: 
\begin{align}
    & \mathrm{Error} = \left| \frac{LB_{\mathrm{QC}} - LB_{\mathrm{exact}}}{LB_{\mathrm{exact}}} \right|,
\end{align}
where $LB_{\mathrm{exact}}$ denotes the objective value obtained by the corresponding classical algorithm and $LB_{\mathrm{QC}}$ denotes the objective value obtained by the quantum-assisted algorithm. This measures the relative deviation of the quantum-assisted solution from the classical benchmark.

As shown in Table~\ref{tab:scalability}, all methods achieve either zero or very small optimality gaps at termination, indicating that both classical and quantum-assisted methods can approximate the optimal robust solution effectively. For smaller instances, the quantum-assisted methods recover the same objective value as the classical methods, resulting in zero solution error. For larger instances, the errors remain small, indicating QCCG is able to produce solutions of comparable quality while using a quantum master solver. 
In addition, QCCG requires less computational time than the classical CCG for all instances. For example, for the $10 \times 10$ instance, the classical CCG algorithm requires $0.345$ seconds, while the QCCG method solves the problem in $48.01$ milliseconds. For the $20 \times 20$ instance, QCCG reduces the number of iterations from $5$ to $4$ and reduces the computational time from $3.339$ seconds to $47.92$ milliseconds. These results suggest that the quantum-assisted master solver can reduce the computational effort required to solve the master problem. Also, the CCG-based methods typically require fewer iterations than the Benders-based methods, while achieving the same or nearly identical objective values.
\begin{table}[t]
\centering
\setlength{\tabcolsep}{4pt}
\caption{Performance comparison of classical and quantum-assisted decomposition algorithms.}
\label{tab:scalability}
\begin{tabular}{cccccc}
\hline
Size & Methods & Optimal Gap(\%) & Error(\%) & Iters & Time (s) \\
\hline

\multirow{4}{*}{$3\!\times\!3$}
& C-Benders  & 0 & \multirow{2}{*}{0} & 6 & 0.088\\
& QC-Benders & 0 &                    & 6 & 0.144 \\
& CCG    & 0 & \multirow{2}{*}{0} & 3 & 0.111 \\
& QCCG   & 0 &                    & 4 & 0.016 \\
\hline

\multirow{4}{*}{$5\!\times\!5$}
& C-Benders  & 0                       & \multirow{2}{*}{0} & 8 & 0.186 \\
& QC-Benders & 0                       &                    & 8 & 0.192 \\
& CCG    & 0                       & \multirow{2}{*}{0} & 3 & 0.117 \\
& QCCG   & $1.2\!\times\!10^{-5}$  &                    & 4 & 0.025 \\
\hline

\multirow{4}{*}{$8\!\times\!8$}
& C-Benders  & 0                      & \multirow{2}{*}{0}    & 8 & 0.208 \\
& QC-Benders & $1.3\!\times\!10^{-5}$ &                       & 8 & 0.224 \\
& CCG    & 0                      & \multirow{2}{*}{4.39} & 3 & 0.195\\
& QCCG   & $1.2\!\times\!10^{-5}$ &                       & 3 & 0.048 \\
\hline

\multirow{4}{*}{$10\!\times\!10$}
& C-Benders  & 0                      & \multirow{2}{*}{0}  & 10 & 0.236\\
& QC-Benders & $1.0\!\times\!10^{-5}$ &                     & 10 & 0.256 \\
& CCG    & 0                      & \multirow{2}{*}{0.4} & 3 & 0.345\\
& QCCG   & $1.2\!\times\!10^{-3}$ &                     & 4  & 0.048 \\
\hline

\multirow{4}{*}{$20\!\times\!20$}
& C-Benders  & 0                       & \multirow{2}{*}{0.48} & 45 & 2.453 \\
& QC-Benders & 0.772                   &                       & 16 & 0.448 \\
& CCG    & 0                       & \multirow{2}{*}{0.34} & 5  & 3.339 \\
& QCCG   & $5.12\!\times\!10^{-4}$ &                       & 4  & 0.048\\
\hline
\end{tabular}
\end{table}
We further analyze the size of the QUBO master problem arising in QCCG. In this framework, continuous variables are discretized and encoded into binary variables, yielding a QUBO formulation solved via a quantum annealer. The QUBO size depends on the number of first-stage variables and the resolution of the binary encoding. Table~\ref{tab:qubo_size} reports the number of binary variables and the corresponding QUBO matrix size for the $8 \times 8$ instance. The results show that the QUBO size increases with the number of facilities and the number of scenarios included in the RMP. As additional uncertainty scenarios are introduced during CCG iterations, the master problem expands, leading to a larger QUBO model.Despite this growth, the quantum solver handles the resulting problem sizes efficiently. Overall, these results demonstrate that the QCCG framework can effectively integrate quantum optimization into decomposition-based robust optimization methods.

\begin{table}[t]
\centering
\caption{Size of the QUBO master problem in QCCG for instance size $8 \times 8$.}
\label{tab:qubo_size}
\begin{tabular}{ccc}
\hline
Iterations & Binary variables & QUBO size \\
\hline
1 & 17 & $17 \times 17$ \\
2 & 81 & $81 \times 81$ \\
3 & 145 & $145 \times 145$ \\
\hline
\end{tabular}
\end{table}

\section{Conclusion} 
\label{sec:Conclusion} 
    This paper developed a hybrid quantum–classical column-and-constraint generation (QCCG) framework for solving 2SRO problems. The proposed approach integrates a QUBO-based quantum master solve with a classical subproblem, preserving the decomposition structure of classical CCG while enabling quantum-assisted optimization in the master stage. Our key contribution is to develop a constraint-preserving QUBO reformulation for the restricted master problem, together with a rigorous bound-adjustment mechanism that accounts for discretization error, penalty-model mismatch, and quantum optimization error. This allows the algorithm to maintain valid lower and upper bounds and ensures solution certification despite inexact quantum master solves. We further showed that the proposed framework generalizes classical CCG and reduces to it when the master problem is solved exactly, thereby retaining its convergence guarantees. Numerical experiments on two-stage robust location–transportation problems demonstrate that the proposed QCCG framework achieves solution quality comparable to classical methods while reducing the computational effort associated with solving large mixed-integer master problems. Thus, our work provides a principled framework for integrating QUBO-based quantum optimization into structured optimization methods with theoretical guarantees.

\bibliographystyle{IEEEtran}
\bibliography{Refs}


\end{document}